\documentclass[a4paper,11pt]{article}
\usepackage{fullpage}
\usepackage[hidelinks]{hyperref}
\usepackage{amsthm,amsmath,amssymb,enumerate}
\usepackage{color}
\usepackage{thm-restate}
\allowdisplaybreaks

\title{Worst-Case Efficient Dynamic Geometric Independent Set} 

\author{Jean Cardinal\thanks{Universit\'{e} Libre de Bruxelles (ULB).
\texttt{\{cardinaljean,johniacono,gregkoumoutsos\}@gmail.com}} 
\qquad 
John Iacono\footnotemark[1] \hskip 3pt
\thanks{Supported by the Fonds de la Recherche Scientifique-FNRS under Grant no MISU F 6001 1.} \hskip 3pt 
\thanks{New York University, USA.} 
\qquad
Grigorios Koumoutsos\footnotemark[1] 
\hskip 3pt 
\thanks{Supported by the Fonds de la Recherche Scientifique-FNRS under a ``Scientific Collaborator'' grant.}\\[1.2ex]
Université libre de Bruxelles (ULB), Belgium}

\newtheorem{theorem}{Theorem}[section]
\newtheorem{definition}[theorem]{Definition}
\newtheorem{lemma}[theorem]{Lemma}

\newtheorem{observation}[theorem]{Observation}
\newcounter{note}[section]

\usepackage{amssymb}
\usepackage{mathtools}
\usepackage{algorithm2e}
\usepackage{cite}
\usepackage{comment}

\newtheorem{fact}[theorem]{Fact}

\DeclareMathOperator{\OPT}{OPT}
\DeclareMathOperator{\MIX}{MIX}

\DeclareMathOperator{\dist}{dist}
\DeclareMathOperator{\IN}{IN}
\DeclareMathOperator{\OUT}{OUT}

\usepackage{enumitem}
\setitemize{itemsep=0.75pt}

\begin{document}

\maketitle

\thispagestyle{empty} 
\begin{abstract}
    We consider the problem of maintaining an approximate maximum independent set of geometric objects under insertions and deletions.
    We present data structures that maintain a constant-factor approximate maximum independent set for broad classes of fat objects in $d$ dimensions, where $d$ is assumed to be a constant, in sublinear \textit{worst-case} update time. 
    This gives the first results for dynamic independent set in a wide variety of geometric settings, such as disks, fat polygons, and their high-dimensional equivalents. 

    Our result is obtained via a two-level approach. First, we develop a dynamic data structure which stores all objects and provides an approximate independent set when queried, with output-sensitive running time. We show that via standard methods such a structure can be used to obtain a dynamic algorithm with \textit{amortized} update time bounds. Then, to obtain worst-case update time algorithms, we develop a generic deamortization scheme that with each insertion/deletion keeps (i) the update time bounded and (ii) the number of changes in the independent set constant. We show that such a scheme is applicable to fat objects by showing an appropriate generalization of a separator theorem. 
    
    Interestingly, we show that our deamortization scheme is also necessary in order to obtain worst-case update bounds:  If for a class of objects our scheme is not applicable, then no constant-factor approximation with sublinear worst-case update time is possible. We show that such a lower bound applies even for seemingly simple classes of geometric objects including axis-aligned rectangles in the plane.
\end{abstract}

\vfill
\pagebreak
\tableofcontents
\thispagestyle{empty} %
\vfill
\pagebreak

\setcounter{page}{1}

\includecomment{onlymain}
\excludecomment{onlyapp}

\begin{onlymain}
\section{Introduction}

Dynamic algorithms for classic NP-hard optimization problems is a quite active research area and major progress has been achieved over the last years  (see~\cite{GKKP17,AAGPS19,BHN19,BK19,BHNW21}). For such problems, typically improved performance can be achieved when their \textit{geometric} versions are considered, i.e., when the input has a certain geometric structure. As a result, a very recent line of research considers dynamic algorithms for well-known algorithmic problems in the geometric setting~\cite{AgSoCG20,HNW20,BCIK20, CMR20,CH21}. In this work, we continue this investigation by considering the geometric version of the Maximum Independent Set problem; some of our results improve upon the previous (very recent) work of Henzinger et al.~\cite{HNW20} and Bhore et al.~\cite{BCIK20} and some others provide the first dynamic data structures for various geometric instances.   

\paragraph{Static Independent Set.} In the maximum independent set problem, we are given a graph $G = (V,E)$ and we aim to produce a subset $I \subseteq V$ of maximum cardinality, such that no two vertices in $I$ are adjacent. This is one of the most well-studied algorithmic problems. It is well-known to be \textsf{NP}-complete and hard to approximate: no polynomial time algorithm can achieve an approximation factor $n^{1-\epsilon}$, for any constant $\epsilon>0$, unless \textsf{P=NP}~\cite{Zuck07,Hastad1999}.  

\paragraph{(Static) Geometric Independent Set.} In the geometric version of the maximum independent set problem, the graph $G$ is the intersection graph of a set $V$ of geometric objects: each vertex corresponds to an object, and there is an edge between two vertices if and only if the corresponding objects intersect. 

Besides the theoretical interest, the study of independent sets of geometric objects such as axis-aligned squares, rectangles or disks has been motivated by applications in various areas such as VLSI design~\cite{HM85}, map labeling~\cite{AKS97,VA99} and data mining~\cite{KMP98,BDMR01}.

For most classes of geometric objects, the problem remains NP-hard. A notable exception is the case where all objects are intervals on the real line (the well-known \textit{interval scheduling} problem), where the optimal solution can be computed in polynomial time using a folklore greedy algorithm \cite{DBLP:journals/siamcomp/Gavril72}. However, even the simplest 2d-generalization to axis-aligned unit squares is NP-hard~\cite{fowler1981optimal}. As a result, most of the related work has focused on developing PTASs for certain versions of geometric independent set. This has been achieved for various types of objects such as squares, hypercubes, fat objects and pseudodisks~\cite{HM85,erlebach2005polynomial,chan2003polynomial,chan2012approximation}.

A substantially harder setting is the when the geometric objects are arbitrary (axis-aligned) rectangles in the plane: Despite the intense interest (see~\cite{aw-asmwir-13,ce-amir-16}), no PTAS is known. Until very recently the best known approximation factor was $O(\log \log n)$~\cite{Parinya09,CW21}; in 2021, several $O(1)$-approximation algorithms were claimed~\cite{M21,GKMM+21}. 

\paragraph{Dynamic Geometric Independent Set.} In the dynamic version of geometric independent set, $V$ is a class of geometric objects (for instance squares in the plane) and we maintain a set $S \subseteq V$ of \textit{active} objects. Objects of $V$ may be inserted in $S$ and deleted from $S$ over time and the goal is to maintain a maximum independent set\footnote{This problem should not be confused by the related, but very different, \textit{maximal} independent set problem, where the goal is to maintain an inclusionwise maximal independent set subject to updates on the edges (and not the vertices). This problem has been studied in the dynamic setting, with a remarkable recent progress after a sequence of breakthrough results~\cite{AOSS18,AOSS19,CZ19,BDHSS19}.} of $S$, while keeping the time to implement the changes (the \textit{update time}) sublinear on the size of the input.

\paragraph{Previous Work.} Most previous work dates from 2020 onwards. The only exception is the work of Gavruskin et al.~\cite{gavruskin2015dynamic} who considered the very special case of intervals where no interval is fully contained in any other interval. The study of this area was revitalized by a paper of Henzinger, Neumann and Wiese in SoCG'20~\cite{HNW20}. They presented deterministic dynamic algorithms with approximation ratio $(1+\epsilon)$ for intervals and $(1+\epsilon)\cdot 2^d$ for $d$-dimensional hypercubes, under the assumption that all objects are contained in $[0,N]^d$ and each of their edges has at least unit length. Their update time was polylogarithmic in both $n$ and $N$, but $N$ could be unbounded in $n$. Furthermore, they showed that no $(1+\epsilon)$-approximation scheme with sublinear update time is possible, unless the Exponential Time Hypothesis fails (this follows almost directly from a lower bound of Marx~\cite{Marx07} for the offline problem). 

Subsequently Bhore et al.~\cite{BCIK20} obtained the first results for dynamic geometric independent set without any assumption on the input. For intervals, they presented a dynamic $(1+\epsilon)$-approximation with logarithmic update time; this was recently simplified and improved by Compton et al.~\cite{CMR20}.  For squares,~\cite{BCIK20} presented a randomized algorithm with expected approximation ratio roughly $2^{12}$ (generalizing to roughly $2^{2d+5}$ for $d$-dimensional hypercubes) with amortized update time $O(\log^5 n)$ (generalizing to $O(\log^{2d+1} n)$ for hypercubes).

\subsection{Our results}
In this work we obtain the first dynamic algorithms for fat objects (formally defined in Section~\ref{s:makemix}) in fixed dimension $d$ with sublinear worst-case update time. The precise bounds on the approximation ratio and the update time depend on the fatness constant and the running time for worst-case range query structures for each family of objects. However, the results follow from a generic approach, applicable to all those classes of objects.

\begin{theorem} \label{t:main}
There exists a deterministic data structure for maintaining a $O(1)$-approximate independent set of a collection of fat objects, with sublinear worst-case update time, that reports $O(1)$ changes in the independent set per update. In particular, it achieves the following approximation ratios and worst-case update times, for any constant $0 < \epsilon \leq 1/4$:

\begin{itemize}
    \item $(4+\epsilon)$-approximation with $O(n^{3/4})$ update time for axis-aligned squares in the plane.
    \item $(5+\epsilon)$-approximation with $\tilde{O}(n^{3/4})$ update time for disks in the plane\footnote{Throughout this paper, $\tilde{O}$ suppresses polylogarithmic factors.}.
    \item $(2^d+\epsilon)$-approximation with $O(n^{1-\frac{1}{2d}})$ update time for $d$-dimensional hypercubes. 
    \item $O(1)$-approximation with $O(n^{1-\frac{1}{2dk}})$ update time for fat objects which are unions of $k$ (hyper)rectangles.
    \item $O(1)$-approximation with $\tilde{O}(n^{1-\frac{1}{d(d+1)}})$ update time for fat simplices in $d$ dimensions.
    \item $O(1)$-approximation with $\tilde{O}(n^{1-\frac{1}{d(d+1)k}})$ update time for fat objects which are unions of $k$ simplices in $d$ dimensions.
    \item $O(1)$-approximation with $\tilde{O}(n^{1-\frac{1}{d+2}})$ update time for balls in $d$ dimensions. 
\end{itemize}
\end{theorem}

This result gives the first dynamic algorithms with sublinear update time for all the aforementioned classes of objects, apart from $d$-dimensional hypercubes. Moreover, for hypercubes our result is the first with sublinear \textit{worst-case update time bounds} without assumptions on the input; it also achieves the same approximation ratio as~\cite{HNW20}, which is the best known for any dynamic setting. 


In fact, it seems hard to significantly improve our result on any aspect: First, for the approximation factor, as mentioned above, Henzinger, Neumann and Wiese~\cite{HNW20} proved that (under the Exponential Time Hypothesis) one cannot maintain a $(1 + \epsilon)$-approximate maximum independent set of hypercubes in $d\geq 2$ dimensions with update time $n^{O((1/\epsilon)^{1-\delta})}$ for any $\delta > 0$; therefore the only potential improvement in the approximation ratio is by small constant factors.
Moreover, the update time we obtain is essentially the time required for detecting intersections between objects in a range query data structure\footnote{In a previous version of this paper~\cite{CIK21arxiv,CIK21}, we erroneously claimed a polylogarithmic update time for squares and hypercubes. We withdrew this claim. Reaching polylogarithmic worst-case update time bounds for these objects would require new ideas, distinct from those given in Section~\ref{s:makedisq}.}.
Fatness of the objects is a sensible condition for achieving such results:
we prove that for (nonfat) rectangles, ellipses, or arbitrary polygons, no dynamic approximation in sublinear worst-case update time is possible. 
Finally, we emphasize the remarkable additional property that the number of changes reported per update is always constant.

To obtain the result of Theorem~\ref{t:main} we develop a generic method to obtain dynamic independent set algorithms and show how to apply it for fat objects. 
Our method uses a combination of several components. The first is a data structure which stores all objects, supports insertions and deletions in sublinear $f(n)$ time and, when queried, returns a $\beta$-approximate independent set with output-sensitive running time. The second main component is a generic transformation of such a data structure into a dynamic algorithm with approximation factor $\beta+\epsilon$ and \emph{amortized} update time $O(f(n))$. This is done by periodically updating the solution using the data structure. Then, we apply a generic \emph{deamortization} scheme, involving what we call a \emph{MIX} function, a generic way to switch ``smoothly'' from one solution to another. We design such a MIX function for fat objects using geometric separators. 

\subsection{Notation and preliminaries}
\paragraph{Notation.} 
In what follows, $V$ is a class of geometric objects, such as squares in the plane, and we consider a finite set $S \subseteq V$.
A subset of $S$ is a \emph{maximum independent set (MIS)} of $S$ if all its objects are pairwise disjoint 
and it has maximum size over all sets with this property. We use $\OPT$ to denote the size of a maximum independent set. We use $n$ to denote $|S|$, unless otherwise specified. 
We say that ${I}$ is a $\beta$-approximate MIS if its size is at least $\beta \OPT$, for a constant $0<\beta\leq 1$ \footnote{Note that while stating the main result in Theorem~\ref{t:main}, we used the convention that the approximation ratio is $>1$. This is done mainly for aesthetical reasons, making the result easier to parse. From now on, we assume $\beta <1$; it is easy to see that this is equivalent to a data structure achieving a $(1/\beta)$-approximation in the language of Theorem~\ref{t:main}.}.

The problem we study is to maintain an approximate MIS set $I$ under a sequence of insertions and deletions in $S$, which we call generically an \emph{update}, starting from an empty set. 
This can be expressed as implementing a single operation $\Delta = \textsc{Update}(u)$, where $u$ is the object to be inserted or deleted, and $\Delta$, the \emph{update set}, is the set of objects that change in the approximate MIS $I$, presented as the symmetric difference from the previous set.
In general we will use subscripts to express variables' states after the indicated update, and unsubscripted versions for the current state.
Using the operator $\oplus$ to denote the symmetric difference, we therefore have $S_i \coloneqq \oplus_{j=1}^{i} \{u_j\}$ and ${I}_i \coloneqq \oplus_{j=1}^{i} \Delta_j$. 
We say that $\textsc{Update}$ is a $\beta$-approximate MIS algorithm if ${I}_i$ is always a $\beta$-approximate MIS of $S_i$.

\paragraph{Comment on the model.} Here, we adopt the convention that the update set $\Delta_i$ is always returned explicitly and thus the running time of an update $u_i$ is at least the size of the returned update set:

\begin{fact}\label{f:zero}
The running time of performing an update is at least $|\Delta_i|$, i.e., the size of the update set.
\end{fact}
We proceed with a simple yet  crucial observation about MIS.
\begin{fact} \label{f:one}
The size of a MIS can change by at most one with every update: $|\OPT_i - \OPT_{i-1}| \leq 1$.
\end{fact}
Note that this fact does not hold for the weighted version of the MIS problem.
Also note that this does not say that it is possible to have an update algorithm with an update set $\Delta_i$ with cardinality always at most $1$; this fact bounds how the size of a MIS can change after an update, and not the content. 
In fact, it is easy to produce examples where $|\Delta_i|$ must be 2, even for intervals (e.g.~$u_1=[1,4]$, $u_2=[2,3]$, $u_3=[1,4]$). 
However, it does leave open the possibility to have an update operation returning constant-size update sets, which is something we will achieve for the classes of geometric objects we consider.

\end{onlymain}

\begin{onlymain}

\section{Overview} \label{s:prelim}

Our goal is thus to develop a general method for dynamic approximate MIS that has efficient worst-case running times and small update set sizes for various classes of geometric objects. We identify two ingredients that are needed: 
\begin{description}
    \item[$\beta$-dynamic independent set query structure ($\beta$-DISQS):] 
    This an abstract data type for maintaining a set $S$ of $n$ objects, in which one can insert or delete an object in time $f(n)$, and obtain a $\beta$-approximate MIS of the current set $S$ in output-sensitive time $kf(n)$ if the set returned has size $k$. 
    \item[MIX algorithm:] A MIX algorithm receives two independent sets $S_1$ and $S_2$ whose sizes sums to $n$ as input, and smoothly transitions from $S_1$ to $S_2$ by adding or removing one element at a time such that at all times the the intermediate sets are an independent sets of size at least $\min(|S_1|,|S_2|)-o(n)$. The running time of the MIX algorithm is said to be $\gamma(n)$ if the entire computation takes time $n\gamma(n)$.
\end{description}
The plan of the paper is as follows.

\paragraph{Section~\ref{s:dyn}: Amortized Update Time.} In this section, we prove that the first ingredient, the $\beta$-DISQS, is sufficient to obtain a $(\beta-\epsilon)$-approximate dynamic MIS algorithm with $O(f(n))$ amortized update time for any $\epsilon>0$. This is presented as Theorem~\ref{th:dyn}. We note that this is a bit of a ``folklore'' algorithm that essentially does nothing more than periodically querying the $\beta$-DISQS.

\paragraph{Section~\ref{s:deam}: Worst-case Update time.}
The second ingredient, MIX, is vital to deamortizing: In Section~\ref{s:deam}, we show that a DISQS and a MIX together are sufficient to produce a data structure which for any constant $\epsilon>0$ maintains an approximate MIS of size at least $(\beta-\epsilon)\OPT-o(\OPT)$, 
has a worst-case running time of $O(f(n)+\gamma(n)+\log n)$, and whose \textsc{Update} operation returns a constant-sized update set; this is presented as Theorem~\ref{th:deam}. We also show in  that the non-existence of a MIX function implies the impossibility of a approximate MIS data structure with sublinear update set size, hence sublinear worst-case running time, and that this impossibility holds for rectangles and other classes of nonfat objects (Theorem~\ref{t:nomix} and Lemma~\ref{l:noapprox}).\\

\noindent Given this generic scheme developed, to prove our main result, it remains to show that fat objects have a MIX function and a DISQS, both with sublinear update time. This is the content of Sections~\ref{s:makemix} and~\ref{s:makedisq}.

\paragraph{Section~\ref{s:makemix}: Existence of a MIX function for fat objects.} 
We show in Lemma~\ref{lem:mix_result} that for fat objects in any constant dimension a MIX algorithm always exists with $\gamma(n)=O(\log n)$. This is achieved via geometric separators~\cite{DBLP:conf/focs/SmithW98}. 

\paragraph{Section~\ref{s:makedisq}: Existence of $\beta$-DISQS for fat objects.} Last, we show that obtaining a $\beta$-DISQS is possible for many types of fat objects using variants of standard range query data structures (kd-trees for orthogonal objects and partition trees for non-orthogonal objects such as disks, triangles or general polyhedra) with the running time $f(n)$ matching the query time of such structures. The result is based on a simple greedy algorithm for the MIS problem on fat objects~\cite{EKNS00,MBHRR95},
yielding an approximation factor $\beta$ that only depends on the fatness constant.
\end{onlymain}

\begin{onlymain}
\section{Dynamization} \label{s:dyn}
In this section we define formally the $\beta$-dynamic independent set query structure ($\beta$-DISQS) and show that its existence implies a dynamic independent set algorithm with approximation ratio $\beta-\epsilon$ and sublinear amortized update time.


\begin{definition}
A \emph{$\beta$-dynamic independent set query structure ($\beta$-DISQS)} is a data structure that maintains a set $S$ whose size is denoted as $n$ and supports the following operations:
\begin{itemize}
    \item $\textsc{Update}(u)$: Insert or remove $u$, so that $S$ becomes $S \oplus u$.
    \item $\textsc{Query}$: Returns a $\beta$-approximate MIS of $S$
\end{itemize}
We say that the running time of a DISQS is $f(n)$ if \textsc{Update} takes time $f(n)$ and if \textsc{Query} returns a set of size $k$ in time $k f(n)$. We require $f(n)$ to be sublinear.
\end{definition}

\paragraph{From $\beta$-DISQS to dynamic with amortized update time.}
We now show that a $\beta$-DISQS is sufficient to give a approximate MIS data structure with an amortized running time, with a loss of only $\epsilon$ in the approximation factor for any $\epsilon>0$. The intuition is simple, pass through the update operations to the DISQS and periodically replace the approximate MIS seen by the user by querying the DISQS. The only subtlety is to immediately remove any items from the approximate MIS that have been deleted in order to keep the approximate MIS valid. This simple transformation is likely folklore, but we work out the details for completeness.

\begin{restatable}{theorem}{tdyn} \label{th:dyn}
Given a $\beta$-DISQS with sublinear running time $f(n)$ for an independent set problem and  $0 < \epsilon < 1$, there is a fully dynamic data structure to maintain a $(1-\epsilon)\cdot\beta$-approximate independent set that runs in  $O\left(\frac{1}{\epsilon}f(n) + \log n \right)$ amortized time per operation. 
\end{restatable}


\begin{proof}
We describe the implementation of $\textsc{Update}(u)$. Initially set global variables $S={I}=\emptyset$, $t=0$, and $D$ is an empty $\beta$-DISQS. 
The sets $I$ and $S$ are stored in a balanced binary search tree, indexed according to some total order on the objects (for instance lexicographically by the coordinates that define them). 
The variables here are defined so that after the $t$th \textsc{Update} operation the invariants $I=I_t$ and $I_{\text{old}}=I_{t-i}$ hold. Pick $\epsilon'$ such that $\frac{1-\epsilon'}{1+\epsilon'} =1 - \epsilon $. It is easy to show that $\epsilon' = \Theta(\epsilon)$. 

\begin{algorithm}[H]
  \SetKwFunction{FMain}{Update}
  \SetKwFunction{FQ}{Query}
  \SetKwProg{Fn}{Function}{:}{}
  \Fn{\FMain{$u$}}{
  D.\FMain{u}\;
  Set $i = i+1$\;
  \uIf{$i = \left\lceil |{I}_{\text{old}}| \epsilon' \right\rceil$\tcc*[r]{periodic rebuild of ${I}$}}
{
        Set ${I}_{\text{new}} = D.\FQ()$ \;
        Set $\Delta = {I}_{\text{new}} \oplus {I}_{\text{old}}$ \;
        Set ${I}={I}_{\text{old}}={I}_{\text{new}}$ \;
        Set $i=0$ \;
}
\Else{
  \uIf{$u \in {I}$\tcc*[r]{$u$ is a delete}}{Set $\Delta=\{u\}$\;
  Remove $u$ from ${I}$}
\uElse{Set $\Delta=\emptyset$}}

\Return $\Delta$ 
}
\end{algorithm}

\noindent 
We first argue that ${I}$ is an independent set of $S$. This is the case after each rebuild, as the query to the DISQS returns an independent set by definition. Between rebuilds, items are only deleted from ${I}$, thus ${I}$ will never contain intersecting items. By deleting items from ${I}$ when they are deleted from $S$, this ensures that ${I}\subseteq S$.

%

\paragraph{Approximation ratio.} 
We now argue that  $|{I}|\geq (1-\epsilon) \cdot \beta \OPT$. To this end we make the following observation.

\begin{observation}
\label{obs_del}
If $i$ updates have been performed since the last rebuild, then we have that $|I| \geq |{I}_{\text{old}}| - i $. 
\end{observation}

This follows since in each update the size of $I$ can shrink by at most 1, which happens only in case $u_t$ is in ${I}$ and must be removed. We have that:


\begin{align*}
\frac{|{I}|}{\OPT}  & \geq \frac{|{I}_{\text{old}}| - i}{\OPT_{\text{old}} + i}    &  \text{Fact~\ref{f:one} and Observation~\ref{obs_del}}\\[1em]
 & \geq \frac{(1-\epsilon') |{I}_{\text{old}}|}{\OPT_{\text{old}} +\epsilon' \cdot|{I}_{\text{old}}|}  & i \leq  \epsilon' \cdot  |{I}_{\text{old}}| \\[1em]
 &\geq \frac{(1-\epsilon') |{I}_{\text{old}}|}{(1+\epsilon')\cdot \OPT_{\text{old}}}    &    |{I}_{\text{old}}|  \leq \OPT_{\text{old}} \\[1em]
 &\geq \frac{(1-\epsilon')\cdot \beta \cdot \OPT_{\text{old}}}{(1+\epsilon')\cdot \OPT_{\text{old}}}  & |{I}_{\text{old}}|  \geq \beta \cdot \OPT_{\text{old}} \\[1em]
 &= \frac{1-\epsilon'}{1+\epsilon'} \cdot \beta = (1-\epsilon) \cdot \beta .    & \text{Since } \frac{1-\epsilon'}{1+\epsilon'} =1 - \epsilon
\end{align*}


\paragraph{Update time.}  Finally, we need to argue that the amortized running time is  $O(\frac{1}{\epsilon}f(n)+\log n)$. 
If an operation is not a rebuild, the running time is simply $f(n)$ for the update in the DISQS and $O(\log n)$ to update the binary search tree containing $S$ (and possibly $I$).

After $\epsilon' \cdot |{I}_{\text{old}}|$ operations, $|{I}|f(n)$ time is spent with the call to update. Thus the amortized cost per operation for the rebuild is:

$$
\frac{|{I}|f(n)}{\epsilon' \cdot |{I}_{\text{old}}|}
\leq
\frac{(|{I}_{\text{old}}|+i)f(n)}{\epsilon' \cdot |{I}_{\text{old}}|} 
 \leq 
\frac{(1+\epsilon')|{I}_{\text{old}}|f(n)}{\epsilon'\cdot |{I}_{\text{old}}|}
 \leq 
\frac{(1+\epsilon')}{\epsilon'}f(n)
$$
Note that $n$ is the current size of the set $S_i$ at the time of the rebuild, but as this can only change by a $(1+\epsilon')$ factor between rebuilds, and $f$ is sublinear, this gives an amortized time of $\frac{(1+\epsilon')}{\epsilon'}(f(n_i)+o(f(n_i))$ for every update $u_i$. 
Thus combining the amortized running time for the rebuild with the $O(\log n)$ for the BST operations and $f(n)$ for the DISQS update, gives the total amortized running time of $O(\frac{1}{\epsilon'}f(n)+\log n) = O(\frac{1}{\epsilon}f(n)+\log n) $.
\end{proof}

\end{onlymain}
\begin{onlymain}

\section{Deamortization} \label{s:deam}
\end{onlymain}
\begin{onlyapp}
\section{Deferred proofs of section \ref{s:deam}}\label{app:deam}
\end{onlyapp}

\begin{onlymain}
In this section we present our deamortization technique. In particular, we describe a procedure, which we call MIX, which if exists, is used to transform a  $\beta$-DISQS into a deterministic dynamic algorithm for with worst-case update time guarantees and update set size bounds. 
We also show that if a MIX does not exist for an independent set problem, then no sublinear worst-case update time guarantees are possible. 

\paragraph{MIX function.} We now define our main ingredient  for deamortization, which essentially says that we can smoothly switch from one solution to another, by adding or removing one item at a time:

\begin{definition}[MIX function]
Given two solution sets $A$ and $B$, let $\MIX(A,B,i)$, for $i \in [0,|A|+|B|]$ be a set where: 
\begin{itemize}
    \item $\MIX(A,B,i)$ is always a valid solution
    \item $\MIX(A,B,0)=A$ 
    \item $\MIX(A,B,|A|+|B|)=B$ 
    \item $|\MIX(A,B,i)| \geq  \min(|A|,|B|)-\Gamma(|A|+|B|)$, for some $\Gamma(|A|+|B|)=o(|A|+|B|)$.
    \item $\MIX(A,B,i)$ and $\MIX(A,B,i+1)$ differ by one item.
\end{itemize}
\end{definition}

Given this purely combinatorial definition, we define a MIX algorithm as follows.

\begin{definition}
\label{def:mix-alg}
A MIX algorithm with running time $\gamma(n)$ is a data structure such that
\begin{enumerate}
    \item It is initialized with $A$, $B$, $i=0$ and it has a single operation \textsc{Advance} which advances $i$ and reports the single element in the symmetric difference $\MIX(A,B,i) \oplus \MIX(A,B,i-1)$.
    \item The initialization plus $|A|+|B|$ calls to \textsc{Advance}  run in total time $(|A|+|B|)\cdot \gamma(|A|+|B|)$.
\end{enumerate}
\end{definition}

\paragraph{Organization.} The rest of this section is organized as follows. In~\ref{sec:deam-alg} we show that given a $\beta$-DISQS and a MIX function for an independent set problem, we can produce a dynamic algorithm with worst-case update time guarantees and approximation ratio arbitrarily close to $\beta$. In~\ref{sec:deam-lb}, we show the necessity of MIX function; in other words, we show that if there does not exist a MIX function for an independent set problem, then no deterministic algorithms with worst-case update time guarantees exist. 

\subsection{Dynamic algorithm with worst-case update time}
\label{sec:deam-alg}

We now present our main theorem.
As we discuss next, the intuition expands upon that for Theorem~\ref{th:dyn}, that in addition to periodically using the $\beta$-DISQS to get a new solution, the MIX slowly transitions between the two previous solutions reported by the $\beta$-DISQS. 
Our result is the following. 

\begin{theorem} \label{th:deam}
Given a $\beta$-DISQS with running time $f(n)$ for an independent set problem, a $\gamma(n)$-time MIX algorithm with nondecreasing 
$\Gamma(n)=o(n)$, and an 
$0 <\epsilon < 1/4$, there is a fully dynamic data structure to maintain an independent set of size at least 
$(\beta-\epsilon)\OPT - o(\OPT) $.
 The data structure runs in  
 $O_{\epsilon,\beta}(f(n)+\gamma(n) + \log n)$ 
 worst-case time per operation, where $n$ is the current number of objects stored, and reports a $O_{\epsilon}(1)$ number of changes in the independent set per update. 
\end{theorem}




\paragraph{High-level description.} 
We saw in the previous section how to obtain an amortized update time algorithm by splitting the update sequence into \emph{rounds} and query the DISQS at the end of each round to recompute an approximate MIS. Let $\hat{I}_k$ be the independent set produced by the DISQS at the end of round $k$. At a high-level, the main task is to deamortize the computation of $\hat{I}_k$: we can not afford computing it during one step. Instead, we compute $\hat{I}_k$ gradually during round $k+1$, making sure that the running time per step is bounded. $\hat{I}_k$ is eventually computed by the end of round $k+1$. At this point, we would like to have $\hat{I}_k$ (discarding elements deleted during round $k+1$) as our output independent set; however this can not be done immediately, since $\hat{I}_k$ might be very different from $\hat{I}_{k-1}$. For that reason, the switch from $\hat{I}_{k-1}$ to $\hat{I}_{k}$ is done gradually using the MIX function during round $k+2$. After all, our algorithm uses $\hat{I}_{k}$ as its independent set at the end of round $k+2$. 

 It follows that the independent set reported to the user is a combination of DISQS queries 3 or 2 rounds in the past. We show that by appropriately choosing the lengths of the rounds depending on the sizes of the independent sets, this lag affects the approximation factor by an additive $\epsilon$.





\begin{proof}[Proof of Theorem~\ref{th:deam}]
We group the updates $u_i$ into rounds, and use $r_k$ to indicate that $u_{r_k}$ is the last update of round $k$. Let $\hat{I}_k$ be the independent set output by the $\beta$-DISQS (if queried) at time $r_k$, i.e., at the end of $k$th round. 
The length of the $k$th round, $R_k = r_k-r_{k-1}$ is defined to be $R_k \coloneqq \max \lbrace 1, \lfloor \epsilon \cdot |\hat{I}_{k-2}|\rfloor \rbrace$ updates. 




%
%
%
For convenience, we define the following functions for any $0<\epsilon<1/4$:
\begin{align*}
& & g(\epsilon) = \frac{1+\sqrt{1-4\epsilon}}{2} & &     h(\epsilon) = \frac{3-\sqrt{1-4\epsilon}}{2}
&&
\phi(\epsilon) = \frac{16 h^2(\epsilon)}{\epsilon \cdot \beta \cdot g^3(\epsilon)}
\end{align*}
Note that $g(\epsilon) = 1-\epsilon-O(\epsilon^2)$ and $h(\epsilon) =1+\epsilon+O(\epsilon^2)$.  Note also that $g(\epsilon) \in (1/2,1)$ and in particular $g(\epsilon) \rightarrow 1$ as $\epsilon \rightarrow 0$. Similarly, $h(\epsilon) \in (1,3/2)$ and $h(\epsilon) \rightarrow 1$ as $\epsilon \rightarrow 0$. 

\medskip
\paragraph{Our Data Structures and Operations.} Our overall structure contains a $\beta$-DISQS and a MIX algorithm.
We also maintain the current active set of objects $S_i$ and our approximate independent set $I_i$ explicitly in a binary search tree; we refer to simply as $S$ and $I$, each stored according to some total order on the objects\footnote{Instead of a tree one could use a hash table, which would remove the additive logarithmic term from the update time, at the expense of randomization. However this will not improve our overall result, since functions $\gamma(n)$ and $f(n)$ are at least logarithmic in our application (see sections~\ref{s:makemix},\ref{s:makedisq}); therefore the logarithm can be absorbed while keeping the result deterministic.}.
We maintain the invariant that at the beginning of round $k$, the DISQS stores the set of objects $S_{r_{k-2}}$. Also, our structure stores $\hat{I}_{k-3}$ and $\hat{I}_{k-2}$. 

To execute update operations in round $k$, the following are performed:

\begin{description}
   \item[Running MIX slowly:] During round $k$ we use MIX to transition from $\hat{I}_{k-3}$ to $\hat{I}_{k-2}$. 
   This is done by initializing MIX with $\hat{I}_{k-3}$ and $\hat{I}_{k-2}$ and repeatedly running the \textsc{Advance} operation.
   After each update, we continue running MIX where we left off and continue until either $\phi(\epsilon) \cdot 
\gamma \left(|\hat{I}_{k-3}|+|\hat{I}_{k-2}|\right)$ time has passed or if
     $\phi(\epsilon)$ calls to \textsc{Advance} have been performed. 
     
    \item[Operation archiving:] All updates are placed into a queue $Q$ as they arrive. This will ensure that the DISQS will take into account all updates of previous rounds. Moreover, if an update $u_i$ deletes an element $v$ of $S$, we wish to report it as being deleted. To do so we set a variable $\Delta_i^{\text{DEL}}=\{v\}$, and otherwise $\Delta_i^{\text{DEL}}=\emptyset$.
    
    \item[Interaction with the DISQS:] During round $k$, we want to use the DISQS in order to compute the set $\hat{I}_{k-1}$. To do that, we first perform to the DISQS all updates of round $k-1$ one by one and remove them from $Q$. This way, DISQS stores the set $S_{r_{k-1}}$. Then, we perform a query to the DISQS, to get $\hat{I}_{k-1}$.
    In each update of the round $\left(1+\frac{2h(\epsilon)}{\beta \epsilon}\right)f(n)$ time is spent on executing these operations.
    
    \item[Maintaining $S$:] We store $S$ in a binary search tree based on some total ordering of the objects. For each update $u_i$, we search for $u_i$ in $S$ and remove it if it is there, and add it if it is not, to maintain $S=S_i=S_{i-1} \oplus \{u_i\}$.


    \item[Output:] After each update, we report the symmetric difference $\Delta$ between previous and current independent set to the user. Let $\Delta_i^{\text{MIX}}$ be the union of the items returned by the \textsc{Advance} operation of the MIX algorithm during the execution of update $u_i$. We combine $\Delta_i^{\text{MIX}}$ and $\Delta_i^{\text{DEL}}$, and before returning, we remove any items that would result in the insertion into $I_i$ of items that are not in $S$. 
    
    
\end{description}

\paragraph{Roadmap.} We need to argue about (i) correctness, (ii) running time, (iii) approximation ratio and (iv) \textit{feasibility} of our algorithm, i.e., that during each round the computation of MIX and DISQS have finished before the round ends. Before this, we will first mention some basic properties that will be helpful in the analysis. 

\paragraph{Basic Properties.} We first mention some basic properties that will be helpful in the analysis. Although the intuition is clear, some of the proofs are long, and we defer them to Appendix~\ref{app:deam}, to keep the proof focused.

Let $\OPT$ be the size of the current optimum. 
First, let us explore the relations between these $\hat{I}_{k}$ and $\OPT_{r_k}$ in various rounds. We use $t$ to denote the current round and $k$ to denote an arbitrary round. First, since the DISQS, outputs a $\beta$-approximate independent set, we have that  
\begin{equation}
\label{eq:ab}
\OPT_{k} \geq   |\hat{I}_{k}|  \geq \beta \OPT_{r_k}.
\end{equation}
We proceed with the following claim.
\begin{restatable}{claim}{rclaim}
\label{claim:relation}
For any round $k$, we have that $g(\epsilon) \cdot \OPT_{r_{k-1}} \leq  \OPT_{r_{k}} \leq h(\epsilon) \cdot  \OPT_{r_{k-1}}$.
\end{restatable}

Note this implies that $(1-\epsilon-O(\epsilon^2)) \cdot \OPT_{r_{k-1}} \leq  \OPT_{r_k} \leq (1+\epsilon+O(\epsilon^2))  \cdot  \OPT_{r_{k-1}}$, i.e., the optimal cost is changing by a bounded amount between consecutive rounds. Proof is in Appendix~\ref{app:deam}. By the same logic, we get a similar claim on the change of the optimal cost between subsequent rounds. At any time during current round $t$ we have that 
\begin{equation} \label{eq:opt}
g^{\ell}(\epsilon) \cdot \OPT_{r_{t-\ell}} \leq  \OPT \leq h^{\ell}(\epsilon) \cdot \OPT_{r_{t-\ell}}
\end{equation}

As a corollary, we get that at any time during current round $t$,

\begin{equation}
\label{eq:S_t}
|\hat{I}_{t-2}|+|\hat{I}_{t-3}| \leq \OPT_{r_{t-2}}+\OPT_{r_{t-3}} 
\leq \frac{\OPT}{g^2(\epsilon)}+\frac{\OPT}{g^3(\epsilon)}
\leq \frac{2\OPT}{g^3(\epsilon)}.
\end{equation}

\paragraph{Correctness.} It is easy to see that the algorithm described above always outputs an independent set to the user: During current round $t$, the user always sees  $\MIX(\hat{I}_{t-3},\hat{I}_{t-2},j)$ for some $j$, with perhaps some items that have been deleted in round $t-1$ or the current round removed. MIX is by definition an independent set at every step as in any subset of it, so the user always sees an independent set.


\paragraph{Bound on update set size.} 
The update set returned is a subset of $\Delta_i^{\text{MIX}} \cup \Delta_i^{\text{DEL}}$, as insertions of $\Delta_i^{\text{MIX}}$ might be removed if the items are no longer in $S$. The size of $\Delta_i^{\text{MIX}}$ is at most $\phi(\epsilon)$ by construction; the set $\Delta_i^{\text{DEL}}$ contains at most one element. Therefore,  

\begin{equation} \label{eq:updateset}
\Delta \leq \phi(\epsilon)+1
= \frac{16 h^2(\epsilon)}{\epsilon \cdot \beta \cdot g^3(\epsilon)} +1
\leq \frac{16}{\epsilon \cdot \beta} \cdot \frac{(3/2)^2}{(1/2)^3} +1
=  
  O\left(\frac{1}{\epsilon \cdot \beta}\right) = O_{\epsilon,\beta}(1)
 \end{equation}

%

\paragraph{Running time.} We now argue about the running time of a single update. 
\begin{itemize}
    \item The running time due to the execution of MIX is $$\phi(\epsilon) \cdot \gamma (|\hat{I}_{t-3}|+|\hat{I}_{t-2}|) \leq \phi(\epsilon) \cdot  \gamma\left(\frac{2\OPT}{g^3(\epsilon)}\right)  \leq \phi(\epsilon) \cdot  \gamma(16 \cdot n) = O_{\epsilon}(\gamma(n)),$$ where the first inequality holds due to~\eqref{eq:S_t}, the second due to  $g^3(\epsilon) \geq (1/2)^3 = 1/8 $ and $\OPT \leq n$  and the equality at the end holds since $\gamma$ is sublinear.
    
    
    \item The running time due to the interaction with the DISQS is $$(1+\frac{2h(\epsilon)}{\beta \cdot \epsilon}) f(n) \leq \frac{4h(\epsilon)}{\beta \cdot \epsilon} \cdot f(n)  = O_{\epsilon,\beta}(f(n))$$
    \item There are also the operations in the binary search trees holding $I$ and $S$, for updating $S$ and checking whether objects output by MIX have been deleted. Those operations cost $O(\log n)$ each (as $|I|\leq|S|\coloneqq n$). The number of such operations is linear in the update set size, which by~\eqref{eq:updateset} is $O_{\epsilon,\beta}(1)$. Therefore their total cost is  $O_{\epsilon,\beta}(\log n)$.
\end{itemize}
\textit{Putting everything together:} Overall, we get that the total update time is upper bounded by $O_{\epsilon,\beta}(f(n)+\gamma(n)+ \log n)$.

\paragraph{Approximation ratio.}  Now we argue the approximation factor. First, we bound the number of items at any step of MIX during current round $t$:
\begin{align*}
& |\MIX(\hat{I}_{t-2},\hat{I}_{t-3},j)| \\  &  \geq \min(|\hat{I}_{t-2}|,|\hat{I}_{t-3}|)-\Gamma(|\hat{I}_{t-2}|+|\hat{I}_{t-3}|)& \text{Defn of MIX}
\\ & \geq \min(\beta \OPT_{i-2},\beta \OPT_{i-3}) 
-\Gamma\left( |\hat{I}_{t-2}|+|\hat{I}_{t-3}|\right)
& \text{From \eqref{eq:ab}}
\\
& \geq \min\left(\frac{\beta}{h^2(\epsilon)} \OPT,\frac{\beta}{h^3(\epsilon)} \OPT \right)
-\Gamma\left(\frac{1}{g^2(\epsilon)} \OPT+\frac{1}{g^3(\epsilon)}  \OPT\right)
& \text{From~\eqref{eq:opt},~\eqref{eq:S_t}}\\
& \geq \frac{\beta}{h^3(\epsilon)} \OPT 
-\Gamma\left(\frac{2}{g^3(\epsilon)}  \OPT\right)
\\
& \geq \frac{\beta}{h^3(\epsilon)} \OPT 
-o(\OPT), & \text{Since $\Gamma(n)=o(n)$}
\end{align*}
The items that have been deleted are from the round $t-1$ or the current round, and thus are at most:
\begin{align*}
 R_t+R_{t-1} &= \max \lbrace 1, \lfloor \epsilon \cdot |\hat{I}_{t-2}|\rfloor \rbrace + \max \lbrace 1, \lfloor \epsilon \cdot |\hat{I}_{t-3}|\rfloor \rbrace \leq \epsilon (|\hat{I}_{t-2}|+|\hat{I}_{t-3}|) + 2
 \\
& \leq \epsilon \left( \frac{1}{g^2(\epsilon)} \OPT+\frac{1}{g^3(\epsilon)} \OPT  \right) + 2 & \text{Due to~\eqref{eq:S_t}}
\\& \leq  \frac{2\epsilon}{g^3(\epsilon)}  \OPT + 2
\end{align*}
Combining, the independent set is of size at least

$$\left( \frac{\beta}{h^3(\epsilon)} - \frac{2\epsilon}{g^3(\epsilon)} \right)\OPT 
 - o(\OPT) $$
For any $\epsilon>0$, by the definition of $g$ and $h$, there is an $\epsilon'$ such that $\frac{\beta}{h^3(\epsilon)} - \frac{2\epsilon}{g^3(\epsilon)}\geq \beta -\epsilon'$. This gives the independent set size as stated in the theorem.

\paragraph{Feasibility: MIX and interaction with the DISQS complete.} We now show that our algorithm is feasible, i.e., the required updates of MIX and DISQS can be performed during a round while maintaining the promised update time. 

\textit{Interaction with DISQS.} First we consider the interaction with the DISQS. Overall, during current round $t$, the interaction with DISQS consists of performing the updates of round $t-1$ and querying the DISQS to output an independent set.   The total running time is therefore at most

\begin{align*}
(|\hat{I}_{t-1}|+R_t)f(n)
&\leq (\OPT_{r_{t-1}}+R_t)f(n)
\\&\leq (h(\epsilon)\OPT_{r_{t-2}}+R_t)f(n)  
& \text{Claim~\ref{claim:relation}}
\\&\leq \left(\frac{h(\epsilon)}{\beta}|\hat{I}_{t-2}|+R_t\right)f(n)
&\text{$\hat{I}_{t-2}$ is a $\beta$-approx of $\OPT_{r_{t-2}}$}
\\&= \left(\frac{h(\epsilon)}{\beta \epsilon}\cdot \epsilon |\hat{I}_{t-2}|+R_t\right)f(n)
\\&\leq \left(\frac{h(\epsilon)}{\beta \epsilon}\cdot 2 R_t+R_t\right)f(n)  & \text{Using $\epsilon |\hat{I}_{t-2}| \leq 2 R_t$}
\\&= \left(1+\frac{2h(\epsilon)}{\beta \epsilon}\right)  f(n) R_t
\end{align*}
As we spend $\left(1+\frac{2h(\epsilon)}{\beta \epsilon}\right)f(n)$ time on this per operation, after $R_t$ operations it will complete.

\textit{The MIX operation completes:} The last part of the proof is to show that indeed during each current round $t$ the MIX operation has enough time to complete and calculate $\hat{I}_{t-2}$.

Our basic logic is as follows: Suppose you have a sequence of operations to execute that has $b$ breakpoints. Then the sequence will be completed after $b+1$ execution chunks. Thus, to show that the MIX operation completes during round $t$, we have to count its breakpoints operation and show that they are $b \leq R_t-1$. 

Recall that in each step we run MIX until either 
$\phi(\epsilon) \cdot \gamma\left(|\hat{I}_{t-3}|+|\hat{I}_{t-2}|\right)$
time has passed or $ \phi(\epsilon)$ items are output by MIX. We treat those two cases separately and then put everything together. In the subsequent calculations we use that $\phi(\epsilon) = \frac{16 h^2(\epsilon)}{\epsilon \cdot \beta \cdot g^3(\epsilon)} \geq 16$, since all of the factors $\frac{1}{\beta},\frac{1}{\epsilon},\frac{h^2(\epsilon)}{g^3(\epsilon)}$ evaluate to at least one, for any $0 < \epsilon < 1/4$. 

\textit{Breakpoints due to the time constraint:} The chunks that are interrupted due to the first (time) constraint have size $\lfloor \phi(\epsilon) \cdot 
\gamma(|\hat{I}_{t-3}|+|\hat{I}_{t-2}|) \rfloor \geq \frac{\phi(\epsilon)}{2} \cdot 
\gamma(|\hat{I}_{t-3}|+|\hat{I}_{t-2}|)$ . 

Recall that by definition of MIX, initializing MIX with $\hat{I}_{t-3}$ and $\hat{I}_{t-2}$ and performing $k$ calls, takes $k \cdot \gamma(|\hat{I}_{t-3}|+|\hat{I}_{t-2}|)$ time. Since $k \leq |\hat{I}_{t-3}|+|\hat{I}_{t-2}| $, using~\eqref{eq:S_t} we get that the total time of MIX is at most $\tau = \frac{2\OPT}{g^3(\epsilon)} \cdot \gamma(|\hat{I}_{t-3}|+|\hat{I}_{t-2}|)$. 

Therefore, the total time of MIX operation is upper bounded by $\tau$ and the size of chunks interrupted due to time constraints is at least $\frac{\phi(\epsilon)}{2} \cdot 
\gamma (|\hat{I}_{t-3}|+|\hat{I}_{t-2}|)$. This implies that the number of breakpoints due to time constraints is at most 

\begin{equation}
\label{eq:time_const}
\frac{\tau}{\frac{\phi(\epsilon)}{2} \cdot 
\gamma(|\hat{I}_{t-3}|+|\hat{I}_{t-2}|)} = \frac{\frac{2\OPT}{g^3(\epsilon)} \cdot \gamma(|\hat{I}_{t-3}|+|\hat{I}_{t-2}|)}{\frac{\phi(\epsilon)}{2} \cdot 
\gamma(|\hat{I}_{t-3}|+|\hat{I}_{t-2}|)} = \frac{4\OPT}{g(\epsilon)^3 \cdot \phi(\epsilon)}
\end{equation}
\textit{Breakpoints due to MIX changes:} The chunks that are interrupted due to the second (MIX) constraint have size $\lfloor \phi(\epsilon) \rfloor \geq \phi(\epsilon)/2$. Note that MIX will output at most  $|\hat{I}_{t-2}|+|\hat{I}_{t-3}| \leq \frac{2\OPT}{g(\epsilon)^3}$ items.  Therefore, the number of breakpoints due to MIX is at most 

\begin{equation}
\label{eq:mix_constr}
\frac{\frac{2\OPT}{g(\epsilon)^3}}{\phi(\epsilon)/2}    = \frac{4\OPT}{g(\epsilon)^3 \cdot \phi(\epsilon)}
\end{equation}
\textit{Putting everything together:} By ~\eqref{eq:time_const} and~\eqref{eq:mix_constr}, we get that the total number of breakpoints is at most 

\begin{align*}
\frac {8\OPT}{g(\epsilon))^3 \cdot \phi(\epsilon)}
&=  \frac{ 8\OPT } { g(\epsilon))^3 \cdot \frac{16 h^2(\epsilon)}{\epsilon \cdot \beta \cdot g^3(\epsilon)}} = \frac{\epsilon \cdot \beta \cdot \OPT}{2 h^2(\epsilon) }
\\&\leq \frac{\epsilon \cdot \beta \cdot \OPT_{t-2}}{2}   & \text{Since $\frac{\OPT}{h^2(\epsilon)} \leq \OPT_{t-2}$ by ~\eqref{eq:opt}}
\\& \leq \frac{\epsilon |\hat{I}_{t-2}|}{2}  & \text{ Using $ \hat{I}_{t-2 }\geq \beta \cdot \OPT_{t-2}$}
\\& \leq R_t & \text{Using $\epsilon \cdot |\hat{I}_{t-2}|\leq 2\cdot R_t$} 
\end{align*}
Thus the mix operation will finish during the round.

\end{proof}

\subsection{Necessity of the MIX function}
\label{sec:deam-lb}


\begin{theorem} \label{t:nomix}
Suppose for any $n$, there are independent sets $A$ and $B$ of size $n$ such there is no
MIX function with $\Gamma(n) \leq (1-\beta)n$ 
for all $A'\subseteq A$ and $B' \subseteq B$, where $|A'|,|B'| \geq \beta n$.
Then there is no $(\beta +\epsilon)$-approximate dynamic MIS algorithm that reports at most $o(n)$ changes in the independent set per update, for any $\epsilon>0$.
\end{theorem}
\begin{proof}
Let $\delta(n)=o(n)$ and suppose there is a $(\beta-\epsilon)$-approximate MIS algorithm that reports $\delta(n)$ changes per update in the worst case.
Insert $A$ in to the data structure. Then insert $B$ and delete $A$. This is $|A|+|B|$ update operations. At all times the independent set is at least $(\beta-\epsilon) n$, and there are at most $\delta(n)$ changes per update operation. We could transform this into a MIX function by taking the at most $\delta(n)$ changes from update and report each change one at a time, first deletes and then inserts; thus at each step of the resultant MIX, their independent set is at least $(\beta + \epsilon) n - \delta(n)$ which is at least $\beta n$ for sufficiently large $n$
\end{proof}

We can apply this to the case of rectangles in the plane, where we show that with a non-trivial worst-case performance of $o(n)$ changes in the independent set per operation,  it is impossible to have an $\beta$-approximate MIS for any $\beta>0$.

\begin{lemma} \label{l:noapprox}
There is no MIX function for rectangles with $\Gamma(n)<n$. 
Thus from Theorem~\ref{t:nomix}, for any $\beta>0$, there is no $\beta$-approximate dynamic MIS algorithm that reports at most $o(n)$ changes in the independent set per update.
\end{lemma}
\begin{proof}
This is equivalent to saying that there are sets of rectangles $A$ and $B$ such that for any MIX function, there is an $i$ such that $\MIX(A,B,i)=\emptyset$.
Consider sets of rectangles $A$ and $B$, each of size $n$, in the form of a grid such that $A$ are horizontally thin and disjoint, $B$ are vertically thin and disjoint, and every rectangle of $A$ intersects all rectangles of $B$. In a MIX function, starting from $A$, one can not add a single element from $B$ until all elements of $A$ have been removed. 
\end{proof}

This construction works for any class of objects for such a generalized "hashtag" construction is possible, which includes any class of shapes which are connected and where for any rectangle, there is a shape in the class that has that rectangle as its minimum orthogonal bounding rectangle. This includes natural classes of shapes without fatness constraints, such as triangles, ellipses, polygons, etc.
\end{onlymain}

\begin{onlyapp}

In this section, we  provide the missing details from the proof of Theorem~\ref{th:deam} from Section~\ref{s:deam}. In particular, we prove Claim~\ref{claim:relation}, which is the only remaining piece towards proving Theorem~\ref{th:deam}.

\paragraph{Basic Properties.} First let us mention some more properties that follow the definition of our algorithm. Recall that  $t$ denotes the current round and $k$ denotes an arbitrary round. Moreover, $\OPT$ denotes the size of the optimal solution at any arbitrary time, during current round $t$.

\begin{align}
    \OPT_{r_k} & \leq  \OPT_{r_{k-1}} + R_k \label{eq:ac} \\
    \OPT_{r_k} & \geq \OPT_{r_{k-1}} -R_k \label{eq:ad}\\
    \OPT & \leq  \OPT_{r_t} + R_t \label{eq:ae} \\
    \OPT & \geq \OPT_{r_t} -R_t \label{eq:af}
\end{align}

The first is by definition of the size of a round. \eqref{eq:ac} and \eqref{eq:ad} follow from Fact~\ref{f:one}, which states that OPT can only change by 1 per operation. The last two hold when $k$ is the most recent round by the same logic. 

\paragraph{Proof of Claim~\ref{claim:relation}.} Now, we are ready to prove Claim~\ref{claim:relation}, which we first restate.

\rclaim*
\begin{proof}
Note that $g(\epsilon) = 1 - \frac{\epsilon}{g(\epsilon)}$ and $h(\epsilon) = 1+ \frac{\epsilon}{g(\epsilon)}$.

\paragraph{First inequality.} To prove the first inequality, $g(\epsilon) \cdot \OPT_{r_{k-1}} \leq  \OPT_i $, we use induction. Recall that $R_k= \lfloor \epsilon \cdot |\hat{I}_{k-2}|\rfloor$. 
Assume that $R_k = \lfloor \epsilon \cdot |\hat{I}_{k-2}|\rfloor$. We have that 
\begin{align*}
\OPT_{r_k} 
&\geq \OPT_{r_{k-1}}-R_k & \eqref{eq:ad}\\
&\geq \OPT_{r_{k-1}}-\epsilon |\hat{I}_{k-2}| & \text{Since $ R_k= \lfloor \epsilon \cdot |\hat{I}_{k-2}|\rfloor$}\\
&\geq \OPT_{r_{k-1}}-\epsilon\OPT_{r_{k-2}} & \eqref{eq:ab}\\
&\geq \OPT_{r_{k-1}}-\epsilon\cdot\frac{\OPT_{r_{k-1}}}{g(\epsilon)} & \text{Induction}
\\
&
\geq \left(1- \frac{\epsilon}{g(\epsilon)} \right)\OPT_{r_{k-1}}
\\
&
= g(\epsilon) \cdot \OPT_{r_{k-1}}  & \text{Since $g(\epsilon) = 1 - \frac{\epsilon}{g(\epsilon)}$}
\end{align*}

Note that in case $R_k = 1$ the inequality holds since $\OPT_{r_k} 
\geq \OPT_{r_{k-1}}-1 \geq g(\epsilon) \OPT_{r_{k-1}}$, given that $\OPT_{r_{k-1}} \geq 2$. 

\paragraph{Second inequality.} We show that the second inequality, $\OPT_{r_{k}} \leq h(\epsilon) \cdot  \OPT_{r_{k-1}}$, follows as a corollary of the first one. 
Again, let us assume that $R_k = \lfloor \epsilon \cdot |\hat{I}_{k-2}|\rfloor$. We have that

\begin{align*}
\OPT_{r_{k}} 
&\leq \OPT_{r_{k-1}}+R_k & \eqref{eq:ac}\\
&\leq \OPT_{r_{k-1}}+\epsilon |\hat{I}_{k-2}| & \text{Since $ R_k= \lfloor \epsilon \cdot |\hat{I}_{k-2}|\rfloor$}\\
&\leq \OPT_{r_{k-1}}+\epsilon\OPT_{r_{k-2}} & \eqref{eq:ab}\\
&\leq \OPT_{r_{k-1}}+\epsilon\cdot\frac{\OPT_{r_{k-1}}}{g(\epsilon)} & \text{Induction}
\\
&
\leq \left(1+ \frac{\epsilon}{g(\epsilon)} \right)\OPT_{r_{k-1}}
\\
&
= h(\epsilon) \cdot \OPT_{r_{k-1}} & \text{Definition of $h(\epsilon)$}
\end{align*}

In case $R_k=1$, then the inequality trivially holds, since $\OPT_{r_k} \leq  \OPT_{r_{k-1}}+1$ which is at most $\OPT_{r_{k-1}} \cdot h(\epsilon)$, given that $\OPT_{r_{k-1}} \geq 2$. 
\end{proof}
\end{onlyapp}
\begin{onlymain}

\section{A MIX algorithm for fat objects}
\label{s:makemix}

In this section we show that our deamortization scheme applies to fat objects, by showing that fat objects have a MIX algorithm with runtime $\gamma(n) = O(\log n)$. 

\paragraph{Fat objects.} 
There are many possible definitions of fat objects in Euclidean space, we use the following one from~\cite{chan2003polynomial}. 
Define the center and size of an object to be the center and side length of one of its minimal enclosing hypercube.

\begin{definition}
\label{defn:fat}
A collection $S$ of (connected) objects is $f$-fat, for a constant $f$, if for any size-$r$ box $R$, there are $f$ points such that every object that intersects $R$ and has size at least $r$ contains one of the chosen points.
\end{definition}

This implies that any box can only intersect $f$ disjoint objects of size larger than the box. 
Throughout the whole section, $f$ and the dimension $d$ are considered to be constant.

\paragraph{Roadmap.} 
We will develop and use a variant of the rectangle separator theorem of Smith and Wormald~\cite{DBLP:conf/focs/SmithW98}. 
We first state the classic version, and then prove the variant we need. 
Our proofs are straightforward adaptations of those in~\cite{DBLP:conf/focs/SmithW98}. 

\begin{lemma}[Smith and Wormald \cite{DBLP:conf/focs/SmithW98}]
For any set $S$ of disjoint squares objects in the plane, there is a separating rectangle $R$ that such
if we partition $S$ into $S^{\IN}$, $S^{\OUT}$ and $S^{ON}$ based on whether each object lies entirely inside $R$, entirely outside $R$, or intersects $R$, $S^{\IN} \leq \frac{2}{3}|S|$,
$S^{\OUT} \leq \frac{2}{3}|S|$ and $S^{ON} = O(\sqrt{|S|})$.
\end{lemma}

What we need differs from this in that we have two sets of fat objects, in each set the objects are disjoint but intersection is allowed between the two sets, and we want to have the separator intersect with an order-square-root number of objects in each set. However we require the separator to be balanced with respect to the first set only; it is not possible to require balance with respect to both sets. We state the separator lemma here; the proof is the subject of Section~\ref{sec:sep}.

\begin{restatable}{lemma}{lfatmix} \label{l:fatmix}
Let $S_1$ and $S_2$ be two sets of disjoint $f$-fat objects in $d$-dimensions. Let $n = |S_1|+ |S_2|$. We can compute a hypercube $s$ and sets $S^{\IN}_1,S^{\IN}_2,S^{\OUT}_1,S^{\OUT}_2$ with the following properties in time $O(d \cdot n) = O(n)$:
\begin{itemize}
    \item The hypercube $s$ intersects $O(n^{1-\frac{1}{d}})$ objects of $S_1 \cup S_2$. 
    \item $S^{\IN}_1 \subseteq S_1$, $S^{\IN}_2 \subseteq S_2$, $S^{\OUT}_1 \subseteq S_1$, $S^{\OUT}_2 \subseteq S_2$
    \item All objects in $S^{\IN}_1$ and $S^{\IN}_2$ lie entirely inside $s$
    \item All objects in $S^{\OUT}_1$ and $S^{\OUT}_2$ lie entirely outside $s$
    \item $|S^{\IN}_1| \leq \frac{4^d-1}{4^d} |S_1|$
    \item $|S^{\OUT}_1| \leq \frac{4^d-1}{4^d} |S_1|$
\end{itemize}
\end{restatable}

Given this separator lemma, we can construct a MIX algorithm with running time $\gamma(n) = O(\log n)$. The main novelty of lemma~\ref{l:fatmix} is that we achieve the separation in linear time; previous works had separator lemmas running in ``polynomial time''. It turns out that this fast running time is the key towards achieving a $O(\log n)$-time MIX algorithm; this will be clear after analyzing our MIX algorithm (see the discussion at the very end of Section~\ref{sec:mix-alg}).

Formally, our main result regarding MIX algorithm for fact objects is the following.

\begin{restatable}{lemma}{mixresult}
\label{lem:mix_result}
Fat objects in constant dimension $d$ have a MIX algorithm with running time $\gamma(n)=O(\log n)$: Given independent sets of fat objects $S_1$ and $S_2$, there is a MIX from $S_1$ to $S_2$ whose size is always at least $\min(|S_1|,|S_2|)-\Gamma(|S_1|+|S_2|)$, with $\Gamma(n)=O( n^{1-\frac{1}{d}}\log n) = o(n)$. The total running time of initializing the algorithm with $S_1$ and $S_2$ and performing all steps of MIX is $O((|S_1|+|S_2|) \cdot \log(|S_1|+|S_2|))$.
\end{restatable}

The proof is the topic of Section~\ref{sec:mix-alg}.

\subsection{Proof of separator lemma}\label{sec:sep}
In this subsection we prove lemma~\ref{l:fatmix}.

We start with a technical lemma.
Let $\square(c,r)$ be a hypercube centered at $c$ and with side length (size) $r$.

\begin{lemma} \label{l:getroot}
Let $S_1$ and $S_2$ be two sets of disjoint $f$-fat objects in $d$-dimensions. Let $n = |S_1|+ |S_2|$. Given two concentric hypercubes $s_1=\square(c,r)$ and $s_2=\square(c,2r)$, there is a hypercube $s_3=\square(c,r'), r\leq r' \leq 2r$ that intersects at most $(f \cdot d \cdot 4^{d}+2) \cdot n^{1-\frac{1}{d}}$ objects of $S_1 
\cup S_2$. Furthermore, $s_3$ can be computed in time $O(d\cdot n)$.
\end{lemma}

\begin{proof}
Fix for each fat object a minimum enclosing hypercube, arbitrarily chosen if it is not unique. For simplicity of presentation we assume general position: all coordinates and sizes of enclosing hypercubes are unique and no enclosing hypercube has size exactly $\frac{r}{\sqrt[d]{n}}$; these assumptions can be removed via standard methods. We classify each object in of $S_1$ and $S_2$ as \emph{big} or \emph{small} based on whether its size is larger or smaller than $\frac{r}{\sqrt[d]{n}}$. 

The proof is divided into two steps. First, the combinatorial part, the existence of hypercube $s_3$ and second, the algorithmic part: finding $s_3$ in linear time.

\smallskip
\paragraph{Combinatorial part: existence of $s_3$.} Consider the $\frac{1}{2}\sqrt[d]{n}$ hypercubes  $s_i=s(c,r+\frac{2ir}{\sqrt[d]{n}})$ for $i \in [1..\frac{1}{2} \sqrt[d]{n}]$. All of these hypercubes are on or in the hyperannulus defined by $s_1$ and $s_2$. We call them \emph{candidate hypercubes}. We will show that at least one of these hypercubes meets the requirements to be the separator $s_3$ of the lemma. To this end, we treat small and big objects separately. 

\textit{Small objects:} Each of the candidate hypercubes is at distance $\frac{r}{\sqrt[d]{n}}$ from each other, and thus each small object in $S_1 \cup S_2$ can only intersect at most one of the $\frac{1}{2}\sqrt[d]{n}$ candidate hypercubes. 
Thus there will be at least one candidate hypercube which intersects at most $\frac{n}{\frac{1}{2}\sqrt[d]{n}}=2 n^{1-\frac{1}{d}}$ small objects from $S_1 \cup S_2$.

\textit{Big objects:}  Now we bound how many big objects from $S_1$, and by symmetry $S_2$, can intersect \textit{any} candidate hypercube $s_i$. To do that, we upper bound the number of disjoint hypercubes of size $r/\sqrt[d]{n}$ that may intersect a candidate hypercube; call this bound $\alpha$. Since each hypercube of size $r/\sqrt[d]{n}$ intersects at most $f$ big objects of $S_1$ (due to Definition~\ref{defn:fat}), we get that at most $f \cdot \alpha$ big objects of $S_1$ might intersect a candidate hypercube. By symmetry, the same bound holds for $S_2$, so overall, $ f \cdot 2 \alpha $ big objects of $S_1 \cup S_2$ may intersect a candidate hypercube. 

We now compute the upper bound $\alpha$. Each candidate hypercube $s_i$ is comprised of $2d$ faces, each of which is a $(d-1)$-dimensional hypercube of size at most $2r$. 
 Covering a size $x$ hypercube in $d$ dimensions with disjoint $\geq d$ dimensional hypercubes of size $y$ can be done with $\left \lceil \frac{x}{y} \right\rceil^d\leq \left( \frac{2x}{y} \right)^d$ size $y$ hypercubes. Since $x \leq 2r$, and using $y = r/\sqrt[d]{n}$, we get that each  $(d-1)$-dimensional face of any $s_i$ can be covered with at most $ \left( \frac{4r}{r/\sqrt[d]{n}}  \right)^{d-1}=4^{d-1}n^{1-\frac{1}{d}}$ hypercubes of size $y$. Since there are $2d$ such facets in each $s_i$, we get a total of at most $ \alpha = 2d \cdot 4^{d-1}n^{1-\frac{1}{d}}$ hypercubes of size $r/\sqrt[d]{n}$ to cover any $s_i$. 

By the discussion above, we get that overall $f \cdot 2 \alpha  =  f \cdot d \cdot 4^dn^{1-\frac{1}{d}}$ big objects in $S_1 \cup S_2$ intersect any $s_i$.
 
\textit{Putting everything together:} 
Combining the discussion on small and big objects, we get that there exists a candidate hypercube that intersects at most $(2 + f \cdot d \cdot 4^d) \cdot n^{1-\frac{1}{d}}$ objects from $S_1 \cup S_2$. 

\smallskip
\paragraph{Algorithmic part: Finding $s_3$ in time $O(d\cdot n)$.} What remains is the algorithmic claim; we know there is a candidate hypercube that has less than $2n^{1-\frac{1}{d}}$ small objects that intersect it, we need only to identify it.
To do so, we compute for each small object the $L_1$ (Manhattan) distance $\dist$ from the center of its enclosing hypercube to $c$. 
Since the points on the candidate hypercubes have $L_1$ distance from $c$ that are multiples of $\frac{2r}{\sqrt[d]{n}}$, rounding $\dist$ to the nearest multiple of $\frac{2r}{\sqrt[d]{n}}$ will identify the only candidate hypercube that it \emph{could} intersect. Thus, for each small object we identify the unique candidate hypercube that it could intersect. Then we choose the candidate hypercube with the fewest such possible intersections; we know from the combinatorial part of the proof that the chosen hypercube will have at most $2n^{1-\frac{1}{d}}$ possible intersections with small objects.

Computing the $L_1$ distances takes linear time, $O(d \cdot n)$, as the $n$ centers are each $d$-dimensional. Finding the hypercube with fewer (possible) intersections takes time $O(\sqrt[d]{n}) = o(n)$. Therefore the overall running time is $O(d \cdot n)$
\end{proof}

We are now ready to prove lemma~\ref{l:fatmix}, which we first restate here.

\lfatmix*

\begin{proof}
For each object in $S_1 \cup S_2$, designate a \emph{representative point}, which is simply an arbitrary point inside the object.
Call the first three dimensions $x,y,z$.
Compute the $|S_1|/4,|S_1|/2,3|S_1|/4$th order statistics of the $x$-coordinates of the representative points of the objects in $|S_1|$, call them $x_1,x_2,x_3$.
Let $[x_a,x_b]$ be the smaller (in size) of $[x_1,x_2]$ and $[x_2,x_3]$. Then looking at the $|S_1|/4$ representative points of the objects in $|S_1|$ with $x\in [x_a,x_b]$, compute the $|S_1|/16,|S_1|/8,3|S_1|/16$ order statistics of the $y$-coordinates, call them $y_1,y_2,y_3$.
Let $[y_a,y_b]$ be the smaller (in size) of $[y_1,y_2]$ and $[y_2,y_3]$.
Continue this process with the $z$ and other dimensions, if needed.
Let $s_1=\square(c,r)$ be a minimal hypercube containing the hyperrectangle $[x_a,x_b]\times [y_a,y_b]\cdots$, and let $s_2=\square(c,2r)$ be the concentric square to $s_1$ with double the size. Then at least $|S_1|/4^d$ representative points are in $s_1$ and $|S_1|/4$ representative points are out of $s_2$ (namely those with $x>x_3$ if $[x_\ell,x_r]=[x_1,x_2]$ and those with $x<x_1$ if $[x_\ell,x_r]=[x_3,x_4]$ ). 

We then apply lemma~\ref{l:getroot} which returns a square $s$ in the annulus of $s_1$ and $s_2$ that intersects at most 
$(f \cdot d \cdot 4^{d}+2)n^{1-\frac{1}{d}} = O(n^{1-\frac{1}{d}})$
objects of $S_1 \cup S_2$. 
Given $s$, $S^{\IN}_1,S^{\IN}_2,S^{\OUT}_1,S^{\OUT}_2$ can be computed.
Thus, the at least $\frac{|S_1|}{4^d}$ objects with centers inside $s$ will not be in $S^{\OUT}_1$, and the at least $\frac{|S_1|}{4^d}$ objects with centers outside $s$ will not be in $S^{\IN}_1$ (This is the point where we need the objects to be connected). This guarantees that each of $S^{\IN}_1,S^{\OUT}_1$ had at most a fraction $\frac{4^d-1}{4^d}$ of the objects of $S_1$.

Lemma~\ref{l:getroot} runs in $O(nd)$ time, as does the additional work described here: computing the $d$ order statistics of a geometrically shrinking set of points \cite{DBLP:journals/jcss/BlumFPRT73}, and classifying each object based on its relationship to $s$.
\end{proof}
\end{onlymain}

\begin{onlymain}
\subsection{MIX algorithm and analysis}
\label{sec:mix-alg}
In this section, we prove lemma~\ref{lem:mix_result}, which we first restate here. 

\mixresult*

\begin{proof}
We compute the separator $s$ from Lemma~\ref{l:fatmix}. Let $S_1^{\IN}$, $S_1^{\OUT}$, $S_1^s$
    $S_2^{\IN}$, $S_2^{\OUT}$, $S_2^s$, denote partition of $S_1$ and $S_2$ into parts that are completely inside the separator, completely outside, and those that intersect the separator, respectively.

\paragraph{Main Idea.}
The main idea is the following: First, we will remove all objects of $S_1^s$. Then the remaining objects of $S_1$ would be either completely inside $s$ or completely outside $S$. We will use recursively the MIX function in both sides to switch from $S_1^{\IN}$ to $S_2^{\IN}$ and from $S_1^{\OUT}$ to $S_2^{\OUT}$ respectively. At the end we will add $S_2^s$. 

\paragraph{Applying the recursion carefully.} Recall that our goal is twofold: First, minimize the running time of MIX, second, make sure that at each step the size of the current set is at large as possible. Formally, minimize both $\gamma(n)$ and $\Gamma(n)$, which should be definitely sublinear, and as small as possible. Note that towards the second goal, i.e., keeping the set size as high as possible at all times, the sets used to apply the first recursive call of MIX matter: should we start from switching from  $S_1^{\IN}$ to $S_2^{\IN}$ or from $S_1^{\OUT}$ to $S_2^{\OUT}$? We want to make the choice that leads to the largest independent set at the intermediate step when only one of the two recursive calls has been applied. 

We denote $a$ and $b$ the sides of separator ($\IN$ and $\OUT$) so that $|S_1^a| +|S_2^b|  \leq |S_1^b| +|S_2^a| $ holds. 
As $\min(w+x,y+z)\leq \max(w+z,x+y)$:
\begin{equation}
    |S_1^b| + |S_2^a| 
    = \max(    |S_1^b| + |S_2^a| ,|S_1^a| +|S_2^b|)
    \geq \min(|S_1^a|+|S_1^b|, |S_2^a|+|S_2^b| ). 
\label{eq:whichdir}
\end{equation}
This way, at the intermediate step when we have mixed only one side, the set size will not be smaller than the beginning or the end of the mix operation.

\paragraph{Formal Description.} The MIX function then proceeds as follows:
\begin{enumerate}[nosep]
    \item \label{step:1} Start with $S_1$.
    \item \label{step:2} Remove the elements of $S_1^s$, one at a time, to give $S_1^a \cup S_1^b$.
    \item \label{step:3} Recursively MIX $S_1^a$ to $S_2^a$. 
    \item \label{step:4} Now we have $S_2^a \cup S_1^b$.
    \item \label{step:5} Recursively MIX $S_1^b$ to $S_2^b$. At the end of this process we will have $S_2^b \cup S_2^b$.
    \item \label{step:6} Add the elements of $S_2^s$, one at a time.
    \item \label{step:7}We finish with $S_2^a \cup S_2^b \cup S_2^s=S_2$.
\end{enumerate}
The base case is when one of the two sets is empty, and the MIX proceeds in the obvious way by deleting all elements of $S_1$ if $S_2$ is empty, or inserting all elements of $S_2$ if $S_1$ is empty. In such a case the lemma holds trivially as $\min(|S_1|,|S_2|)$ is zero.

We need to argue that at all times this process generates a set that is an independent set of the claimed size. 

\paragraph{Always an independent set.} In steps \ref{step:1}-\ref{step:2} we always have a subset of $S_1$, which is an independent set, the same holds in steps \ref{step:6}-\ref{step:7} with respect to $S_2$.
In step~\ref{step:4},  $S_2^a \cup S_1^b$ is independent as each of the sets are independent and all of the objects on each set are entirely on opposite sides of the separating rectangle $s$. Steps~\ref{step:3} and \ref{step:5} hold by induction, and that the part we are recursively MIXing and the part that is unchanged are entirely on opposite sides of the separating rectangle.

\paragraph{Size bound.}
Let $\MIX_{\min}(S_1,S_2,n)$ be the smallest size of the independent set during the running of the MIX function from $S_1$ to $S_2$, and where $n$ is an upper bound on $|S_1|+|S_2|$. Then we can directly express $\MIX_{\min}$ as a recurrence, taking the minimum of each step of the algorithm:
\begin{align*} \label{eq:mixrec}
    \MIX_{\min}(S_1,S_2,n)=\min(
    \begin{split}
    |S_1|,
    |S_1^a|+|S_1^b|,
    \MIX_{\min}(|S_1^a|,|S_2^a|,n)+|S_1^b|,
    |S_2^a| +|S_1^b| ,\\
    |S_2^a|+\MIX_{\min}(|S_1^b|,|S_2^b|,n),
    |S_2^a|+|S_2^b|,
    |S_2|)
    \end{split}
\end{align*}
We will prove that
\begin{equation}
    \label{eq:mix_size}
    \MIX(S_1,S_2,n)\geq \min(|S_1|,|S_2|)-(\log_{\frac{4^d}{4^d-1}} |S_1|) \cdot  n^{1-\frac{1}{d}}
\end{equation}
which implies the claim of this lemma. The details of this inductive proof are given as Lemma~\ref{l:mix_size} in Appendix~\ref{ax:mix}. 
It follows exactly the intuition that (i) the only loss in the MIX function is the separators at each level of the recursion, (ii) there are logarithmic levels of the recursion and (iii) the size of the separators in each level are $O(n^{1 - 1/d})$. Some care must be taken, since the separators are only balanced for one of the two sets.

\end{onlymain}

\begin{onlyapp}

\section{Missing proofs from Section~\ref{s:makemix}}\label{ax:mix}

Here we prove the size claim of lemma~\ref{lem:mix_result} 

Recall that $\MIX_{\min}(S_1,S_2,n)$ be the smallest size of the independent set during the running of the MIX function from $S_1$ to $S_2$, and where $n$ is an upper bound on $|S_1|+|S_2|$. Then we can directly express $\MIX_{\min}$ as a recurrence, taking the minimum of each step of the algorithm:
\begin{equation} \label{eq:mixrectwo}
    \MIX_{\min}(S_1,S_2,n)=\min
    \begin{cases}
    |S_1|\\
    |S_1^a|+|S_1^b|\\
    \MIX_{\min}(|S_1^a|,|S_2^a|,n)+|S_1^b|\\
    |S_2^a| +|S_1^b| \\
    |S_2^a|+\MIX_{\min}(|S_1^b|,|S_2^b|,n)\\
    |S_2^a|+|S_2^b|\\
    |S_2|
    \end{cases}
\end{equation}
We will prove that
\begin{lemma}
    \label{l:mix_size}
    $$\MIX(S_1,S_2,n)\geq \min(|S_1|,|S_2|)-(\log_{\frac{4^d}{4^d-1}} |S_1|) \cdot  n^{1-\frac{1}{d}}$$
\end{lemma}
\begin{proof}

We need to prove this by induction for each of the cases of~\eqref{eq:mixrectwo}, but two cases dominate the others:
As 
\begin{align*}
|S_1|\geq |S_1^a|+|S_1^b| &\geq \MIX_{\min}(|S_1^a|,|S_2^a|,n)+|S_1^b|\\
    |S_2^a| +|S_1^b| &\geq \MIX_{\min}(|S_1^a|,|S_2^a|,n)+|S_1^b| \\
|S_2| \geq |S_2^a|+|S_2^b| &\geq |S_2^a|+\MIX_{\min}(|S_1^b|,|S_2^b|,n)\\
    |S_2^a| +|S_1^b| &\geq |S_2^a|+\MIX_{\min}(|S_1^b|,|S_2^b|,n)
\end{align*}
it is sufficient to show: \begin{align*}
   \MIX_{\min}(|S_1^a|,|S_2^a|,n)+|S_1^b|
&\geq \min(|S_1|,|S_2|)-(\log_{\frac{4^d}{4^d-1}}|S_1|) n^{1-\frac{1}{d}} & \text{and} 
\\
|S_2^a|+\MIX_{\min}(|S_1^b|,|S_2^b|,n) &
\geq \min(|S_1|,|S_2|)-(\log_{\frac{4^d}{4^d-1}}|S_1|) n^{1-\frac{1}{d}}. 
\end{align*}

\paragraph{Case 1:} Proving $\MIX_{\max}(S_1^a,S_2^a,n)+S_1^b \geq \min(|S_1|,|S_2|)-(\log_{\frac{4^d}{4^d-1}} |S_1|) n^{1-\frac{1}{d}}$
\begin{align*}
&\MIX_{\max}(S_1^a,S_2^a,n)+|S_1^b | 
\\
&\geq
\min(|S_1^a|,|S_2^a|)-(\log_{\frac{4^d}{4^d-1}} |S_1^a|) n^{1-\frac{1}{d}}+|S_1^b|
& \text{Induction}
\\ &
=\min(|S_1^a|+|S_1^b|,|S_2^a|+|S_1^b|)-(\log_{\frac{4^d}{4^d-1}} |S_1^a|) n^{1-\frac{1}{d}}
\\ 
&\geq \min(|S_1^a|+|S_1^b|,|S_2^a|+|S_2^b|)
-(\log_{\frac{4^d}{4^d-1}} S_1^a) n^{1-\frac{1}{d}}
&\text{From \eqref{eq:whichdir}}
\\
&\geq \min(|S_1|^a+|S_1^b|,|S_2^a|+|S_2^b|)
-(\log_{\frac{4^d}{4^d-1}} \frac{4^d-1}{4^d}S_1) n^{1-\frac{1}{d}}
& |S_1^a| \leq \frac{4^d-1}{4^d}S_1
\\
&= \min(|S_1^a|+|S_1^b|,|S_2^a|+|S_2^b|)
-(\log_{\frac{4^d}{4^d-1}} S_1) n^{1-\frac{1}{d}}
+n^{1-\frac{1}{d}}
\\
&= \min(|S_1^a|+|S_1^b|+n^{1-\frac{1}{d}},|S_2^a|+|S_2^b|+ n^{1-\frac{1}{d}})
-(\log_{\frac{4^d}{4^d-1}} |S_1|) n^{1-\frac{1}{d}}
\\
&\geq \min(|S_1^a|+|S_1^b|+ n^{1-\frac{1}{d}},|S_2^a|+|S_2^b|+ n^{1-\frac{1}{d}})
-(\log_{\frac{4^d}{4^d-1}} |S_1|) n^{1-\frac{1}{d}}
&
n\geq |S_1|+|S_2|
\\
&\geq  \min(|S_1^a|+|S_1^b|+|S_1^s|,|S_2^a|+|S_2^b|+|S_2^s|)
-(\log_{\frac{4^d}{4^d-1}} |S_1|) n^{1-\frac{1}{d}}
&
|S_1^s| + |S_2^s| \leq n^{1-\frac{1}{d}}
\\
& = \min(|S_1|,|S_2|)-(\log_{\frac{4^d}{4^d-1}} |S_1|)n^{1-\frac{1}{d}}
\end{align*}

\paragraph{Case 2:} Proving $|S_2^a|+\MIX_{\min}(|S_1^b|,|S_2^b|,n) 
\geq \min(|S_1|,|S_2|)-(\log_{\frac{4^d}{4^d-1}}|S_1|)n^{1-\frac{1}{d}}$

This is almost symmetric. The first three lines differ, and the fourth line is the same as in the first case.

\begin{align*}
&|S_2^a|+\MIX_{\min}(|S_1^b|,|S_2^b|,n) \\
&\geq
|S_2^a| + \min(|S_1^b|,|S_2^b|)-(\log_{\frac{4^d}{4^d-1}} |S_1^b|) n^{1-\frac{1}{d}}
& \text{Induction}
\\ &
=\min(|S_1^b|+|S_2^a|,|S_2^a|+|S_2^b|)-(\log_{\frac{4^d}{4^d-1}} |S_1^a|) n^{1-\frac{1}{d}}
\\ &
=\min(|S_1^a|+|S_1^n|,|S_2^a|+|S_2^b|)-(\log_{\frac{4^d}{4^d-1}} |S_1^a|) n^{1-\frac{1}{d}}
&\text{From \eqref{eq:whichdir}}
\end{align*}
\end{proof}
\end{onlyapp}
\begin{onlymain}

\paragraph{Running time.} The running time is given by the recursion:
$$
T(S_1,S_2)=\begin{cases}
|S_1| & \text{if }S_2=\emptyset
\\
|S_2| & \text{if }S_1=\emptyset
\\
T(S^{a}_1,S^{a}_2)+T(S^{b}_1,S^{b}_2)+ O(|S_1|+|S_2|)
&\text{otherwise}
\end{cases}
$$
where the additive term in the last case is dominated by the time needed to compute the separator (Lemma~\ref{l:fatmix}) which is $O(d \cdot (|S_1|+|S_2|)) = O(|S_1|+|S_2|)$, since $d$ is assumed to be constant.

Recall $|S^{a}_1|,|S^{a}_2|\leq \frac{4^d}{4^d-1}S_1$,  
$S^{a}_1$ and $S^{b}_1$ are disjoint subsets of $S_1$, 
and
$S^{a}_2$ and $S^{b}_2$ are disjoint subsets of $S_2$. Hence the recursion depth is logarithmic in $|S_1|$ and each item from $S_1$ and $S_2$ is passed to at most one of the recursive terms.

Thus the running time is $O((|S_1|+|S_2|)\log |S_1|)$. As the running time is defined to be
$(|S_1|+|S_2|) \cdot \gamma(|S_1|+|S_2|)$, we have $\gamma(|S_1|+|S_2|)\leq \log (|S_1|+|S_2|)$ \end{proof}

\smallskip
\paragraph{Remark.} We note the effect of the running time of the separator algorithm (Lemma~\ref{l:fatmix}) on the running time of the MIX algorithm: If the running time was $O(n^c)$ for some constant $c$, then the additive term in the recursion would have increase to $O((|S_1|+|S_2|)^c)$, leading to $\gamma(n) = n^{c-1} \cdot \log n$; such a running time would be sublinear only for $c <2$; here, by achieving $c=1$, we get the fastest possible running time which implies the logarithmic running time for MIX for fat objects. 

\end{onlymain}

\begin{onlymain}

\section{DISQS for fat objects} \label{s:makedisq}

In this section, we define DISQS for various classes of geometric objects. 

\paragraph{Warm-up: Intervals.}
First observe that for intervals, a $1$-DISQS with running time $O(\log n)$ follows from the classic greedy algorithm \cite{DBLP:journals/siamcomp/Gavril72}. 
By storing the intervals in an augmented binary search tree, one can insert and delete intervals as well as answer queries of the form ``What is the interval entirely to the right of $x$ with the leftmost right endpoint?'' As intervals are fat, this implies a $(1-\epsilon)$ approximate MIS algorithm with running time $O(\frac{1}{\epsilon}\log n)$. 
This in not new, in the past year a complicated structure via local exchanges appeared in~\cite{BCIK20} and soon after a much simpler method \cite{CMR20} using a local rebuilding was obtained; we obtain this as part of our more general scheme for fat objects.

\paragraph{Fat objects.} We now focus on developing DISQS for fat objects. This involves two ideas. First, we use a simple greedy offline algorithm that computes a constant-approximate MIS for fat objects. Then we combine this algorithm with a range searching data structure to implement the greedy choice.

\paragraph{A greedy offline approximation algorithm.}
Given a collection of fat objects, we consider the independent set obtained by sorting them by size, from the smallest to the largest, and adding them greedily to the independent set, provided they do not intersect anything added previously. We refer to this algorithm as the \emph{greedy} algorithm. It was considered in particular by Chan and Har-Peled \cite{chan2012approximation}, but was known in special cases before~\cite{EKNS00,MBHRR95}. 
From the definition of $f$-fat objects, every successive object returned by the algorithm can intersect at most $f$ disjoint objects that are larger.
Hence this simple algorithm yields a constant-factor approximation algorithm for fat objects.
\begin{lemma}
\label{lem:greedy}
For $f$-fat objects in dimension $d$, the greedy algorithm returns an $(1/f)$-approximate MIS.
\end{lemma}

\paragraph{Creating DISQS based on the greedy algorithm.}
We need to implement this greedy method as the query of a DISQS; that is, the running time of the greedy algorithm must be output-sensitive and the DISQS should support insertions and deletions of objects.

\textit{High-level description:} This can be done with a slight variation of classic range intersection query structures, where we can insert and delete objects, mark or unmark objects that intersect a given query, and return the smallest unmarked object. Thus each item returned by the greedy algorithm is reported after a constant number of range intersection query operations. The DISQS data structures thus have running times that match those of the underlying range intersection query structures when the query ranges are from  the same family of objects as the objects stored: $O(n^{1-\frac{1}{2d}})$ for hyperrectangles, $\tilde{O}(n^{1-\frac{1}{d+2}})$ for disks and $\tilde{O}(n^{1-\frac{1}{d(d+1)}})$ for simplices.

\textit{Formal Description:} We now formalize the ideas sketched above. Recall that we wish to implement a DISQS where the query method implements the greedy algorithm in an output sensitive way.


We show how a $O(1)$-DISQS can be obtained, provided there exists a dynamic data structure that stores a collection of geometric objects and allows to (i) \emph{mark} objects that are intersected by a given query object, and (ii) return the smallest unmarked object.
More precisely, the set of operations is the following:
\begin{itemize}
    \item $\textsc{Insert}$ / $\textsc{Delete}$
    \item $\textsc{Unmark-all}$: set all objects as unmarked.
    \item $\textsc{Mark-intersecting}(x)$: Given an object $x$ in the class, mark all those stored that intersect the object.
    \item $\textsc{Smallest-unmarked}$: Report the smallest unmarked object.
\end{itemize}

Note that the type of object passed to $\textsc{Mark-intersecting}$ is in the same class as the objects stored, that is if the structure stores squares, we are interested in marking all squares that intersect a given square. We call such a data structure an \emph{Augmented range query structure (ARQS)}. We refer to $\zeta(n)$ as the running time of the ARQS if all operations run in time $\zeta(n)$, where $n$ is the current number of objects stored.

For fat objects, an ARQS immediately gives a $O(1)$-DISQS with the same running time:
\begin{lemma}
If there is a ARQS with running time $\zeta(n)$ to maintain a class of geometric objects, then there is a $O(1)$-DISQS with running time $O(\zeta(n))$ for these objects.
\end{lemma}
\begin{proof}
The update operation on the DISQS can be passed on to the update of the ARQS. 
The query operation boils down to implementing the greedy algorithm as follows: 
\begin{itemize}
    \item $\textsc{Unmark-all}$
    \item While there are unmarked objects:
    \begin{itemize}
        \item Let $x$ be the object returned by $\textsc{Smallest-unmarked}$
        \item $\textsc{Mark-intersecting}(x)$
    \end{itemize}
    \item Return all successive objects returned by $\textsc{Smallest-unmarked}$
    \end{itemize}
    Correctness follows from Lemma~\ref{lem:greedy}.
\end{proof}



All that remains is to show that there are ARQS for various classes of objects of interest.
Given a dynamic range query structure that supports objects and queries of the given type, augmenting this data structure with markings to get an ARQS is usually an exercise, which we now overview.


Such structures can be easily built using textbook range query structures, lifted into an appropriate dimension and augmented to handle the marking in standard ways. We overview how this can be done, and do not attempt to introduce additional complexity in order to optimize logarithmic factors: the textbook structures that we use have the distinct advantage of producing tree-shaped structures from which marks and minimum-size-in-subtree can be trivially maintained; for more complicated orthogonal log-shavers such as fractional cascading \cite{DBLP:journals/algorithmica/ChazelleG86,DBLP:journals/algorithmica/ChazelleG86a} and its dynamic counterpart \cite{DBLP:journals/algorithmica/MehlhornN90} as well as more advanced non-orthogonal techniques doing this would be less straightforward.

\paragraph{Rectangles.}

We begin with rectangles, showing that
there is a ARQS for rectangles with running time $O(n^{3/4})$. Since arbitrary rectangles are not fat, this is used to get a dynamic independent set algorithm just for squares. 

Use a four-dimensional kd-tree \cite{B75} to store the rectangles as points based on the coordinates of the four sides. Finding rectangles that intersect a given rectangle can be done with a constant number of orthogonal range queries in the 4-dimensional space; $d$-dimensional kd-trees support such queries in time $O(n^{1-\frac{1}{d}})$, therefore $O(n^{3/4})$ for $d=4$. Based on standard dynamization techniques \cite{OL81b,Ov83}, kd-trees can be made dynamic with $O(\log^2 n)$ worst-case update time. 

Each four-dimensional point is augmented with a mark and the size of the square. As a kd-tree is just a tree and we can easily augment it in the usual way (see for instance Chapter~14, \emph{Augmenting Data Structures}, of CLRS~\cite{DBLP:books/daglib/0023376}) to answer the particular queries we need: some nodes of the tree may have \emph{markers} which indicate whether the subtree should be marked or unmarked, with the mark of a node being determined by the highest ancestor with a marker; every time a node is touched, markers are pushed to the children. Additionally each node indicates smallest unmarked item in its subtree. This augmentation can easily be maintained at no additional asymptotic cost during an update.


\paragraph{Orthogonal hyperrectanges.} 

The results for rectangles can be extended to $d$ dimensions using $2d$-dimensional kd-trees (representing each hyperrectangle as a point in $2d$-dimensions, specifying the start and end points in each dimension) with a running time of $O(n^{1-\frac{1}{2d}})$. Again, finding hyperrectangles that intersect a given hyperractangle can be done with a constant (depending on $d$) number of othogonal range queries in the $2d$-dimensional space. Augmentation is done in the same way as for rectangles. Any orthogonal shape of constant complexity shares the same running time by decomposition. This gives our update time for $d$-dimensional hypercubes. 

\paragraph{General polygons and simplices.}

For simplices of constant complexity, there is an ARQS with running time $\tilde{O}(n^{1-\frac{1}{d(d+1)}})$. This follows directly from the standard halfspace range searching literature, applied to simplices, which are defined by $d(d+1)$ coordinates ($d+1$ vertices of $d$ dimensions each). 

Each simplex is represented as a point in the $d(d+1)$-dimensional space and detecting intersections between simplices is equivalent to a constant (depending on $d$) number of halfspace range queries. Use, e.g., the partition tree of Matousek~\cite{M92} which answers halfspace range queries in $D$ dimensions in time $O(n^{1-\frac{1}{D}} \log^{O(1)} n) = \tilde{O}(n^{1-\frac{1}{D}})$; for $D = d (d+1)$ this evaluates to $\tilde{O}(n^{1-\frac{1}{d(d+1)}})$. This structure is presented in~\cite{M92} to support insertions and deletions in amortized update time $O(\log^{2} n)$, but using standard techniques (see, e.g.,~\cite{Ov83}) it can be transformed to a new structure that achieves worst-case update time $O(\log^{2} n)$. Moreover, this method yields a tree-shaped structure, which can be easily augmented to support the marking scheme needed.

We conclude that for $d$-dimensional simplices there is an ARQS with running time $\tilde{O}(n^{1-\frac{1}{d(d+1)}})$.  

\paragraph{Unions of objects.}

By projecting into higher dimensions, these results can be extended to classes of fat objects that are the union of a constant number of simplices or hyperrectanges, such as polygons. This gives running times of  
$O(n^{1-\frac{1}{2dk}})$ for unions of $k$ hyperrectanges and 
$\tilde{O}(n^{1 - \frac{1}{d(d+1)k}})$
for unions of $k$ simplices, assuming $d$ and $k$ are constant.

\paragraph{Disks and balls.}

Methods for disks and other algebraic curves entail lifting and applying the results for non-orthogonal range searching. In~\cite{DBLP:conf/swat/GuptaJS94} it was shown that if the objects and queries are both $d$-dimensional balls, then one can lift into $d+2$ dimensions and apply the halfspace results. Thus there is an AQRS with running time $\tilde{O}(n^{1-\frac{1}{d+2}})$, which is $\tilde{O}(n^{\frac{3}{4}})$ for disks.\\

\paragraph{Summary.} The range intersection queries combined with the greedy algorithm (Lemma~\ref{lem:greedy}) gives a DISQS whose running time depends on the range intersection queries, and whose approximation radio is the inverse of the fatness constant. In particular, the ARQS described above imply the following $\beta$-DISQS for different families of fat objects:

\begin{itemize}
    \item $ (1/4)$-DISQS with running time $O(n^{3/4})$  for axis-aligned squares in the plane.
    \item $ (1/2^d)$-DISQS with running time $O(n^{1-\frac{1}{2d}})$  for $d$-dimensional hypercubes. 
    \item $\Omega(1)$-DISQS with running time $\tilde{O}(n^{1-\frac{1}{d(d+1)}})$ for fat simplices in $d$ dimensions.
    \item $\Omega(1)$-DISQS with running time $O(n^{1-\frac{1}{2dk}})$  for fat objects which are unions of $k$ hyperrectangles in $d$ dimensions.
     \item $\Omega(1)$-DISQS with running time $\tilde{O}(n^{1-\frac{1}{d(d+1)k}})$  for fat objects which are unions of $k$ simplices in $d$ dimensions.
    \item $\Omega(1)$-DISQS with running time $\tilde{O}(n^{1-\frac{1}{d+2}})$  for balls in $d$ dimensions. In particular for disks in the plane, this gives a $(1/5)$-DISQS with running time $\tilde{O}(n^{3/4})$. 
\end{itemize}
From Lemma~\ref{lem:mix_result} fat objects have a MIX algorithm with running time $O(\log n)$. Given the MIX, DISQS, and constant $\epsilon >0$, Theorem~\ref{th:deam} yields a dynamic MIS structure, with worst-case update time depending on the range intersection queries, approximation ratio within $\epsilon$ of that from the fatness, and with only a constant-size update set per operation. Putting all these pieces together yields Theorem~\ref{t:main}.
\end{onlymain}

\section*{Acknowledgments}
We would like to thank Timothy Chan and Qizheng He for pointing out an error in a previous version of this paper. 
\newpage
\bibliographystyle{plain}
\bibliography{bib}

\newpage
\excludecomment{onlymain}
\includecomment{onlyapp}
\appendix

\end{document}